%% file: main.tex
\documentclass[letterpaper,10pt,conference]{ieeeconf}
\usepackage{amsmath,amssymb,mathtools,booktabs,cite}
\usepackage{graphicx}
\usepackage{amsmath,amssymb,mathtools,booktabs,cite}
\usepackage[hidelinks]{hyperref}
\IEEEoverridecommandlockouts
\newtheorem{proposition}{Proposition}
\title{\LARGE \bf
Structured Koopman-Lifted Finite-Memory Identification via Truncated Gr\"unwald--Letnikov Kernels
}

\author{Navid Mojahed$^{1*}$, Mahdis Rabbani$^{1*}$, and Shima Nazari$^{1}$%
\thanks{$^{*}$These authors contributed equally to this work.}
\thanks{$^{1}$The authors are with the Department of Mechanical and Aerospace Engineering, University of California, Davis, CA 95616, USA. 
        {\tt\small nmojahed@ucdavis.edu}}%
}

\begin{document}

\maketitle
\thispagestyle{empty}
\pagestyle{empty}

%%%%%%%%%%%%%%%%%%%%%%%%%%%%%%%%%%%%%%%%%%%%%%%%%%%%%%%%%%%%%%%%%%%%%%%%%%%%%%%%
\begin{abstract}

Nonlinear hereditary systems combine nonlinear state dependence with memory. Koopman lifting provides linear predictors for nonlinear dynamics, but standard formulations are Markovian, while unrestricted lag-based extensions require progressively more fitted coefficients as the retained memory horizon increases. This work proposes a structured Koopman finite-memory model that represents history through a shared normalized Grünwald–Letnikov temporal profile, while learning how the history couples the lifted coordinates directly from data. This structure retains explicit history dependence while making the number of fitted coefficients independent of the retained memory horizon. For fixed memory parameters, identification remains a linear regression problem with closed-form least-squares and ridge solutions. We derive deterministic bounds that distinguish structured-memory approximation, omitted history, regressor conditioning, and regularization effects, and construct an exact augmented Markovian realization for recursive prediction and stability analysis. Numerical studies on a nonlinear viscoelastic robotic system demonstrate the complementary roles of nonlinear lifting and explicit memory, while the structured representation requires substantially fewer fitted coefficients than unrestricted lag-based Koopman models.

\end{abstract}

\section{Introduction}
\label{sec:intro}

Many engineering systems exhibit dynamics that are both nonlinear and
history dependent. Examples arise in viscoelastic structures, compliant
mechanisms, frictional interfaces, dielectric materials, and soft robotic
systems, where internal relaxation processes evolve over time and influence
the measured response
\cite{BagleyTorvik1983,RossikhinShitikova2010AMR,
Caputo2001DistributedOrder,Xun2023FOFriction,RusTolley2015SoftRobots}.
When these internal variables are not measured explicitly, the observable
state generally does not form a Markov state by itself, and the effect of the
unresolved dynamics appears through the history of the measured variables
\cite{Lin2021MoriZwanzig,Catalasan2026SoftRobot}. Accurate prediction in this
setting therefore requires a model that can represent both nonlinear
state dependence and hereditary effects.

Koopman-based methods provide an attractive framework for representing
nonlinear dynamics through linear evolution in a space of observables
\cite{Mezic2005SpectralProperties,Rowley2009KoopmanModes}. Data-driven
approximations such as dynamic mode decomposition (DMD), extended dynamic mode
decomposition (EDMD), DMD with control, and Koopman predictors with control
have consequently become useful tools for nonlinear prediction and control
\cite{Proctor2016DMDc,WilliamsKevrekidisRowley2015EDMD,
KordaMezic2018KoopmanMPC,AbtahiAraghiMojahedNazari2025DeepBilinearKoopman}.
Standard finite-dimensional formulations, however, propagate the current
lifted state as a Markovian state. If the measured variables omit internal
states responsible for memory, nonlinear lifting of the instantaneous state
alone does not recover the missing hereditary dependence.

A number of approaches have introduced history into DMD- and Koopman-type
models. These include delay-coordinate and Hankel embeddings
\cite{Takens1981,ArbabiMezic2017HankelDMD,Brunton2017HAVOK,Kamb2020TimeDelay},
lagged evolution models such as DMD with memory
\cite{Anzaki2023DMDm}, fractional-order DMD
\cite{CaoChen2024FODMD}, and Mori--Zwanzig closure formulations
\cite{Lin2021MoriZwanzig,Catalasan2026SoftRobot,Gupta2025MZAE}.
Related work has also considered trajectory-space operators and episodic
memory
\cite{Rosenfeld2021OKHS,Redman2025EpisodicMemory}.
These methods demonstrate the importance of retaining past information, but a
basic parameterization issue remains for long-memory prediction. If an
independent lifted operator is identified at every retained lag, the number of
fitted coefficients grows linearly with the memory horizon \cite{Anzaki2023DMDm}. For a lifted state
of dimension $p$, an unrestricted $N$-lag model requires $Np^2$ state-history
coefficients, in addition to the input and intercept terms. This growth can
substantially enlarge the identification problem and can aggravate
conditioning and data requirements when long histories are retained.

The key idea of this work is to avoid identifying the lag operators
independently. Instead, we impose a common temporal structure across the
retained lags and learn from data how that memory couples the lifted
coordinates. In particular, the lag matrices are constrained to share a
single temporal profile while their common matrix-valued coupling remains
fully learned. This structured representation retains explicit history
dependence but makes the number of fitted coefficients independent of the
retained memory horizon. 

To parameterize the shared temporal profile, we use a truncated and normalized
Gr\"unwald--Letnikov (GL) sequence. Fractional-order models have long been used
to describe hereditary phenomena because they represent nonlocal temporal
dependence through compact kernel descriptions
\cite{Podlubny1999FDE,ChenPetrasXue2009Tutorial,
mojahed2025fractionalcalculusoptimalcontrol}. Related numerical work has also
studied finite-memory and compressed approximations of such nonlocal operators
\cite{MacDonald2015AdaptiveMemory,GuglielmiHairer2025}. 

The contributions of this work are threefold. First, we introduce a structured finite-memory Koopman model whose number of fitted coefficients is independent of the retained memory horizon. Second, we characterize the approximation induced by the structured memory representation, derive deterministic least-squares and ridge identification bounds, and obtain an exact augmented Markovian realization for recursive prediction and stability analysis. Third, we evaluate the proposed framework on a nonlinear viscoelastic robotic benchmark, examining finite-horizon prediction accuracy, parameter efficiency, and data requirements relative to memoryless and unrestricted lag-based alternatives.

\noindent \textbf{Notation}
For vectors, $\|\cdot\|_2$ denotes the Euclidean norm; for matrices,
$\|\cdot\|_2$ and $\|\cdot\|_F$ denote the spectral and Frobenius norms.
The Moore--Penrose pseudoinverse of $A$ is denoted by $A^\dagger$,
$I_r$ denotes the $r$-dimensional identity, and $\rho(A)$ denotes the
spectral radius.

\section{Structured Lifted Finite-Memory Model}
\label{sec:model}

Consider uniformly sampled measurements
$x_k\in\mathcal{X}\subset\mathbb{R}^n$ and inputs
$u_k\in\mathcal{U}\subset\mathbb{R}^m$ at sampling interval $h>0$.
The measured state may omit internal variables responsible for hereditary
effects and therefore need not be Markovian. We introduce a fixed lifting
map
\begin{equation}
z_k=\psi(x_k)
=
\begin{bmatrix}
x_k\\
\varphi(x_k)
\end{bmatrix}
\in\mathbb{R}^p,
\qquad
x_k=Cz_k,
\label{eq:lift}
\end{equation}
where $\varphi:\mathcal{X}\rightarrow\mathbb{R}^{p-n}$ contains nonlinear
observables and $C=[\,I_n\;\;0\,] \in \mathbb{R}^{n\times p}$. 

To represent history without identifying an independent operator at every
lag, we impose a shared temporal structure on the delayed lifted states.
For $\alpha\in(0,1)$, let the Gr\"unwald--Letnikov (GL) coefficients satisfy
\begin{equation}
w_0(\alpha)=1,
\qquad
w_j(\alpha)
=
\left(1-\frac{1+\alpha}{j}\right)w_{j-1}(\alpha),
\label{eq:gl_raw}
\end{equation}
for $j\ge1$, and for a retained memory horizon $N\ge2$, define
\begin{equation}
c_N(\alpha)
=
\left(\sum_{j=2}^{N}w_j^2(\alpha)\right)^{1/2},
\quad
\widetilde w_j(\alpha,N)
=
\frac{w_j(\alpha)}{c_N(\alpha)} .
\label{eq:gl_normalized}
\end{equation}
Thus $\sum_{j=2}^{N}\widetilde w_j^2=1$. The retained history is summarized by
\begin{equation}
s_k
=
\sum_{j=2}^{N}
\widetilde w_j z_{k+1-j}
\in\mathbb{R}^p .
\label{eq:history}
\end{equation}

The proposed GL-M model is
\begin{equation}
z_{k+1}
=
b+A_0z_k+Bu_k-Ms_k+d_k,
\label{eq:glm}
\end{equation}
where
$b,d_k\in\mathbb{R}^p$,
$A_0,M\in\mathbb{R}^{p\times p}$, and
$B\in\mathbb{R}^{p\times m}$.
% Equivalently, its strictly delayed operators satisfy
% \begin{equation}
% K_j=-\widetilde w_j M,
% \qquad 2\le j\le N .
% \label{eq:lagoperators}
% \end{equation}
To make the structural reduction explicit, consider the corresponding
unrestricted finite-memory model
\begin{equation}
z_{k+1}
=
b+A_0z_k+Bu_k
+\sum_{j=2}^{N}K_j z_{k+1-j}
+d_k ,
\label{eq:freelag_general}
\end{equation}
where each $K_j\in\mathbb{R}^{p\times p}$ is identified independently.
GL-M instead constrains these delayed operators to the form
\begin{equation}
K_j=-\widetilde w_j M,
\qquad 2\le j\le N .
\label{eq:lagoperators}
\end{equation}
For fixed $(N,\alpha)$, the coefficients
$\{\widetilde w_j\}_{j=2}^{N}$ are completely determined by the shared GL
profile and are not fitted independently. Hence, rather than learning a
different $p\times p$ matrix at every lag, GL-M learns a single matrix $M$
that determines how the shared temporal history acts across the lifted
coordinates. GL-M therefore contains $2p^2+pm+p$ fitted scalar coefficients, whereas an unrestricted horizon-$N$ lag model contains $Np^2+pm+p$.

\section{Identification and Model Selection}
\label{sec:identification}

For fixed $(N,\alpha)$, the normalized GL weights are known, and the
history coordinate $s_k$ in \eqref{eq:history} can be constructed directly
from the data. Let $\{k_\ell\}_{\ell=1}^{L}$ denote the indices of the
available transitions and define
\begin{equation}
\begin{aligned}
Z &=
\begin{bmatrix}
z_{k_1} & \cdots & z_{k_L}
\end{bmatrix},
&
Z_+ &=
\begin{bmatrix}
z_{k_1+1} & \cdots & z_{k_L+1}
\end{bmatrix},\\
U &=
\begin{bmatrix}
u_{k_1} & \cdots & u_{k_L}
\end{bmatrix},
&
S &=
\begin{bmatrix}
s_{k_1} & \cdots & s_{k_L}
\end{bmatrix}.
\end{aligned}
\label{eq:data_matrices}
\end{equation}
Using \eqref{eq:glm}, these data satisfy $Z_+ = \Theta\Omega + D$, where
\begin{equation*}
\Theta=
\begin{bmatrix}
b & A_0 & B & -M
\end{bmatrix},
\quad
\Omega=
\begin{bmatrix}
\mathbf{1}_L &
Z^\top &
U^\top &
S^\top
\end{bmatrix}^\top,
\label{eq:regression}
\end{equation*}
where $D$ collects the one-step residuals and
$\Theta\in\mathbb{R}^{p\times q}$ with
$q=1+2p+m$. Thus, once the temporal profile is fixed, identification of
the structured finite-memory model reduces to a linear regression whose
regressor dimension is independent of the retained memory horizon $N$.

The ordinary least-squares estimate is
\begin{equation}
\widehat{\Theta}_0
=
\arg\min_{\Theta}\|Z_+-\Theta\Omega\|_F^2
=
Z_+\Omega^\dagger .
\label{eq:ols}
\end{equation}
This estimate is unique when
$\operatorname{rank}(\Omega)=q$, which necessarily requires $L\ge q$.
Hence, identifiability for a fixed temporal profile is determined by the
rank of the structured regressor $\Omega$.

To regularize poorly conditioned or rank-deficient regressions, we use
\begin{equation}
\widehat{\Theta}_\lambda
=
\arg\min_{\Theta}
\left\{
\frac{1}{L}\|Z_+-\Theta\Omega\|_F^2
+
\lambda\|\Theta R^{1/2}\|_F^2
\right\},
\label{eq:ridge}
\end{equation}
where $\lambda>0$, and
\[
R=\operatorname{diag}(0,I_p,I_m,I_p)
\]
leaves the intercept unpenalized. Defining
\[
G_\lambda
=
\frac{1}{L}\Omega\Omega^\top+\lambda R,
\]
the corresponding closed-form solution is
\begin{equation}
\widehat{\Theta}_\lambda
=
\frac{1}{L}Z_+\Omega^\top G_\lambda^{-1}.
\label{eq:ridge_solution}
\end{equation}
For $\lambda>0$, $G_\lambda$ is positive definite under this
regularization, so the ridge estimate remains unique even when
$\Omega$ is rank deficient, although
regularization does not make unexcited coefficient directions identifiable
from the data.

The remaining quantities $(N,\alpha,\lambda)$ are treated as model-selection
parameters and are chosen from finite candidate sets using recursive
prediction error on held-out validation trajectories. This criterion is
used because one-step regression error does not, in general, determine
multistep prediction performance. At $N=2$,
$\widetilde w_2=-1$ independently of $\alpha$, so $\alpha$ is not
identifiable at this horizon. The final test trajectories are not used in
either fitting or model selection.

\section{Structured Approximation}
\label{sec:approximation}

The GL-M model replaces the independently varying lag operators of an
unrestricted lag representation by the shared-profile constraint in \eqref{eq:lagoperators}. We next quantify the approximation introduced by
this restriction and its effect on the identified coefficients.

For this purpose, consider a prescribed unrestricted lag representation
along the trajectories of interest,
\begin{equation}
z_{k+1}
=
b^\star+A_0^\star z_k+B^\star u_k
+\sum_{j=2}^{\infty}K_j^\star z_{k+1-j}
+r_k^\star ,
\label{eq:reference}
\end{equation}
where $K_j^\star\in\mathbb{R}^{p\times p}$ are unrestricted lag operators
and $r_k^\star\in\mathbb{R}^p$ is the residual of the reference
representation. This has the same unrestricted lag structure introduced
in Section~\ref{sec:model}, but allows an infinite history so that
mismatch over the retained lags can be distinguished from truncation
beyond the horizon $N$. The superscript $^\star$ denotes fixed reference quantities used for comparison, not optimal solutions. A finite-memory reference is included by
setting $K_j^\star=0$ beyond its horizon.

Assume that, along the trajectories under consideration,
\begin{equation}
\sum_{j=2}^{\infty}\|K_j^\star\|_2<\infty,
\quad
\|z_k\|_2\le Z_{\max},
\quad
\|r_k^\star\|_2\le\bar r ,
\label{eq:reference_assumptions}
\end{equation}
for finite constants $Z_{\max}$ and $\bar r$. No exact finite-dimensional
Koopman closure or exact additive memory representation is assumed.

We compare \eqref{eq:reference} with a GL-M representation with any fixed coupling
matrix $M^\star\in\mathbb{R}^{p\times p}$. Define
\begin{equation}
\begin{aligned}
\varepsilon_N
&=
\sum_{j=2}^{N}
\left\|K_j^\star+\widetilde w_jM^\star\right\|_2
+
\sum_{j=N+1}^{\infty}
\left\|K_j^\star\right\|_2,\\
\delta_N
&=
Z_{\max}\varepsilon_N+\bar r .
\end{aligned}
\label{eq:mismatch}
\end{equation}
The first term in $\varepsilon_N$ measures how closely the shared GL
structure represents the unrestricted lag operators over the retained
history, whereas the second measures the omitted memory beyond $N$.
Thus, $\delta_N$ accounts for retained-profile mismatch, memory
truncation, and the residual of the reference representation.

\begin{proposition}
\label{prop:error}
Let $\Theta^\star=
\begin{bmatrix}
b^\star & A_0^\star & B^\star & -M^\star
\end{bmatrix},$
and define the residual of the corresponding structured representation as $d_k^\star
=
z_{k+1}
-
\Theta^\star
\begin{bmatrix}
1 & z_k^\top & u_k^\top & s_k^\top
\end{bmatrix}^{\top}$.

Under \eqref{eq:reference_assumptions},
\begin{equation}
\|d_k^\star\|_2\le\delta_N .
\label{eq:structured_residual_bound}
\end{equation}
For the $L$ identification samples, let $D=Z_+-\Theta^\star\Omega$.
Then
\[
\|D\|_F\le\sqrt{L}\,\delta_N .
\]
If $\Omega$ has full row rank, the least-squares estimate satisfies
\begin{equation}
\|\widehat{\Theta}_0-\Theta^\star\|_F
\le
\frac{\sqrt{L}\,\delta_N}
{\sigma_{\min}(\Omega)} .
\label{eq:ols_bound}
\end{equation}
For $\lambda>0$, no rank assumption is required and
\begin{equation}
\begin{aligned}
\|\widehat{\Theta}_\lambda-\Theta^\star\|_F
\le {}
\frac{\delta_N}{\sqrt{L}}
\left\|\Omega^\top G_\lambda^{-1}\right\|_2+
\lambda
\left\|\Theta^\star R G_\lambda^{-1}\right\|_F .
\end{aligned}
\label{eq:ridge_bound}
\end{equation}
\end{proposition}

\begin{proof}
Subtracting the structured representation associated with
$\Theta^\star$ from \eqref{eq:reference} gives
\begin{equation*}
\begin{aligned}
d_k^\star
={}
\sum_{j=2}^{N}
\left(K_j^\star+\widetilde w_jM^\star\right)z_{k+1-j}+
\sum_{j=N+1}^{\infty}
K_j^\star z_{k+1-j}
+r_k^\star .
\end{aligned}
\label{eq:residual_decomposition}
\end{equation*}
Using $\|z_k\|_2\le Z_{\max}$ and the triangle inequality yields
\eqref{eq:structured_residual_bound}, and therefore
$\|D\|_F\le\sqrt{L}\,\delta_N$.

If $\Omega$ has full row rank,
\[
\widehat{\Theta}_0-\Theta^\star
=
D\Omega^\dagger ,
\]
and \eqref{eq:ols_bound} follows from
$\|\Omega^\dagger\|_2=1/\sigma_{\min}(\Omega)$.
For the ridge estimate, \eqref{eq:ridge_solution} gives
\[
\widehat{\Theta}_\lambda-\Theta^\star
=
\frac{1}{L}D\Omega^\top G_\lambda^{-1}
-
\lambda\Theta^\star R G_\lambda^{-1}.
\]
Applying the corresponding norm inequalities gives
\eqref{eq:ridge_bound}.
\end{proof}

Increasing the retained horizon reduces the omitted-memory term, but does not necessarily reduce the structured mismatch or improve the identified predictor.

\section{Augmented Realization and Prediction}
\label{sec:realization}

The GL-M model in \eqref{eq:glm} is not Markovian in $z_k$ alone because
its evolution depends on a finite history of lifted states. Since this
history is finite, however, the delayed states can be collected into an
augmented state, yielding an equivalent finite-dimensional Markovian
realization. Define
\begin{equation}
\zeta_k=
\begin{bmatrix}
z_k\\
z_{k-1}\\
\vdots\\
z_{k-N+1}
\end{bmatrix}
\in\mathbb{R}^{pN},
\qquad
E=
\begin{bmatrix}
I_p\\
0
\end{bmatrix}
\in\mathbb{R}^{pN\times p},
\label{eq:aug_state}
\end{equation}
and let $F_N=
\begin{bmatrix}
A_0 &
-\widetilde w_2M &
\cdots &
-\widetilde w_NM
\end{bmatrix}
\in\mathbb{R}^{p\times pN}$.
The corresponding augmented matrices are
\begin{equation}
\begin{aligned}
\mathcal A_N
&=
\begin{bmatrix}
F_N\\
I_{p(N-1)} & 0
\end{bmatrix},\quad
\mathcal B_N=EB,\\
b_N&=Eb, \qquad \qquad \quad
\mathcal C_N=
\begin{bmatrix}
C & 0
\end{bmatrix}.
\end{aligned}
\label{eq:aug_matrices}
\end{equation}
The GL-M model in \eqref{eq:glm} is therefore equivalent to
\begin{equation}
\zeta_{k+1}
=
\mathcal A_N\zeta_k
+\mathcal B_Nu_k
+b_N
+Ed_k,
\quad
x_k=\mathcal C_N\zeta_k .
\label{eq:augmented}
\end{equation}
This equivalence applies to the GL-M surrogate rather than to
the underlying physical system. The augmented state has dimension $pN$,
although the number of fitted GL-M coefficients is independent of $N$.

For recursive prediction, the identified GL-M model is propagated without
the unknown future residual, with the augmented state initialized from the
measured lifted history:
\begin{equation}
\widehat\zeta_{k+1}
=
\widehat{\mathcal A}_N\widehat\zeta_k
+\widehat{\mathcal B}_Nu_k
+\widehat{b}_N,
\qquad
\widehat x_k=\mathcal C_N\widehat\zeta_k .
\label{eq:recursive_prediction}
\end{equation}

Let $e_k=\zeta_k-\widehat\zeta_k$ denote the augmented prediction error and
define the one-step residual of the fitted model as
\[
\widehat d_k
=
z_{k+1}
-
\widehat\Theta
\begin{bmatrix}
1 & z_k^\top & u_k^\top & s_k^\top
\end{bmatrix}^{\top}.
\]
For the same input sequence, the prediction error evolves according to
\begin{equation}
e_{k+1}
=
\widehat{\mathcal A}_N e_k
+
E\widehat d_k .
\label{eq:error_dynamics}
\end{equation}
Thus, recursive prediction error is determined by both the one-step model
residual and its propagation through the augmented dynamics.

Because previous prediction errors are repeatedly propagated by
$\widehat{\mathcal A}_N$, the stability of the augmented realization is
relevant to long-horizon prediction. The homogeneous predictor is
asymptotically stable if and only if $\rho(\widehat{\mathcal A}_N)<1$.
This condition is not enforced by the least-squares or ridge identification
and must therefore be assessed separately.

\section{Numerical Results}

\subsection{Benchmark and Experimental Setup}
\begin{figure}[t]
    \centering
    \includegraphics[
        width=0.95\columnwidth,
        trim={6cm 9cm 15cm 2cm},
        clip
    ]{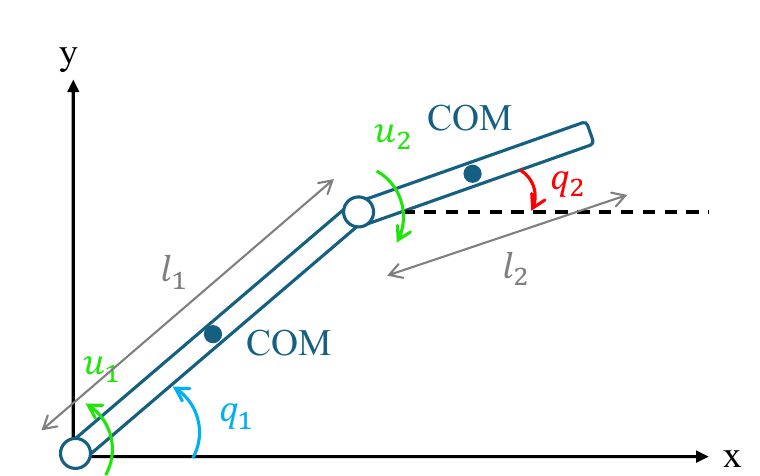}
    \caption{Schematic of the nonlinear two-link manipulator used in the
    numerical study. The measured configuration is described by the joint
    angles $q_1$ and $q_2$, with joint torques $u_1$ and $u_2$ applied at
    the two joints.}
    \label{fig:two_link}
\end{figure}
We consider a nonlinear planar two-link manipulator with hidden
viscoelastic memory, as shown in Fig.~\ref{fig:two_link}. Let $q=[q_1\;q_2]^\top$ denote the joint angles and
$u\in\mathbb{R}^2$ the applied joint torques. The physical dynamics are
\begin{equation*}
\begin{aligned}
    D_q(q)\ddot q
&+c(q,\dot q)+g(q)+D_v\dot q+K_e q
+\gamma\sum_{\ell=1}^{3}\xi_\ell
=u,
\label{eq:plant}\\
\dot\xi_\ell
&=
-\frac{1}{\tau_\ell}\xi_\ell
+A_\ell\dot q,
\qquad \ell=1,2,3,
\end{aligned}
\end{equation*}
where $\xi_\ell\in\mathbb{R}^2$ are internal relaxation states. Only
\[
x=
\begin{bmatrix}
q^\top & \dot q^\top
\end{bmatrix}^{\top}
\in\mathbb{R}^4
\]
is measured. The six relaxation states are hidden from the identification
procedure, so the complete physical realization has ten states while the
identified models observe only four.

Both links have length $0.5$ m, masses $(1.5,1.0)$ kg, center-of-mass
distances $0.25$ m, and uniform-link inertias. The inertia matrix is
\begin{equation}
D_q(q)
=
\begin{bmatrix}
a+2b_0\cos q_2 & d+b_0\cos q_2\\
d+b_0\cos q_2 & d
\end{bmatrix},
\label{eq:inertia_matrix}
\end{equation}
where
\[
a=0.458333,\qquad
b_0=0.125,\qquad
d=0.083333.
\]
The Coriolis and gravitational torques are
\begin{align}
c(q,\dot q)
&=
b_0\sin q_2
\begin{bmatrix}
-2\dot q_1\dot q_2-\dot q_2^2\\
\dot q_1^2
\end{bmatrix},
\\
g(q)
&=
g_0
\begin{bmatrix}
0.875\sin q_1+0.25\sin(q_1+q_2)\\
0.25\sin(q_1+q_2)
\end{bmatrix},
\end{align}
with $g_0=9.81$ m/s$^2$. The viscoelastic parameters are
\[
K_e=\operatorname{diag}(1.5,0.6),
\qquad
D_v=\operatorname{diag}(0.8,0.07),
\]
and
\[
\begin{aligned}
(A_1,A_2,A_3) &= (0.50,0.35,0.15)K_e,\\
(\tau_1,\tau_2,\tau_3) &= (0.15,0.8,2.5)\ {\rm s}.
\end{aligned}
\]

The hidden memory is generated by three exponential relaxation modes
rather than by a GL kernel. This provides a deliberately mismatched test
case in which the temporal structure imposed by GL-M does not coincide
with the true physical memory law. Weak, nominal, and strong memory
regimes are obtained using $\gamma\in\{0.25,1,2\}$,
respectively.

For identification, the state-based models use $z=x$ with $p=4$. The
Koopman models use the fixed $p=14$ lifting
\begin{align}
\psi(x)=\big[&
q_1,q_2,\dot q_1,\dot q_2,
\sin q_1,\sin q_2,\sin(q_1+q_2),\cos q_2,\nonumber\\
&
\dot q_1^2,\dot q_1\dot q_2,\dot q_2^2,
\sin q_2\dot q_1^2,
\sin q_2\dot q_1\dot q_2,
\sin q_2\dot q_2^2
\big]^\top .
\label{eq:dictionary}
\end{align}
The dictionary contains geometric and velocity-dependent nonlinear
observables and does not use the hidden relaxation states.

Five predictors are considered. State-Markov and Koopman-Markov are
memoryless models in the original and lifted coordinates, respectively.
State-GL-M and Koopman-GL-M use the proposed structured memory
parameterization, while Koopman-Free-Lag assigns a separate coefficient matrix to each
retained lag. Their fitted coefficient counts are $[28, \, 238, \, 44, \, 434, \, 196N+42]$, respectively.

Data are generated under feedback tracking of prescribed joint-space
trajectories, with additional excitation included to enrich the
identification data. 
For each memory regime, 360 trajectories are used for training, 90 for
validation, and 150 for testing. To enable matched comparisons across memory
strengths, the same underlying scenarios are simulated separately for
$\gamma=0.25$, $1$, and $2$.

The applied torques are zero-order held at a sampling interval of
$h=0.02$ s. Each simulation includes a 15 s warm-up followed by 500
recorded transitions. The final 400 states preceding the recorded interval
are retained so that memory-based models can be initialized with measured
history. 

All model selection is performed using only the validation trajectories.
The candidate memory horizons are
\[
N\in\{25,50,100,200,300,400\}.
\]
For Koopman-GL-M and State-GL-M, $(N,\alpha,\lambda)$ are selected jointly
from their respective candidate grids. Koopman-Free-Lag selects
$(N,\lambda)$, while the Markov models select only the ridge parameter
$\lambda$. Selection minimizes open-loop validation Normalized
Root--Mean--Square Error (NRMSE) at $H=100$, corresponding to a
two-second prediction horizon.

Fig.~\ref{fig:alpha} illustrates this selection procedure for
Koopman-GL-M. For each fixed memory horizon $N$, the validation NRMSE is
shown as a function of $\alpha$, with the ridge parameter $\lambda$
reselected at every $(N,\alpha)$ pair. The final model is selected by
comparing the resulting validation errors across all candidate values of
$N$ and $\alpha$. Thus, the figure shows how the preferred temporal
profile depends jointly on the retained memory horizon and the GL order,
rather than treating either parameter as fixed a priori.

\begin{figure*}[t]
\centering
\includegraphics[width=\textwidth]{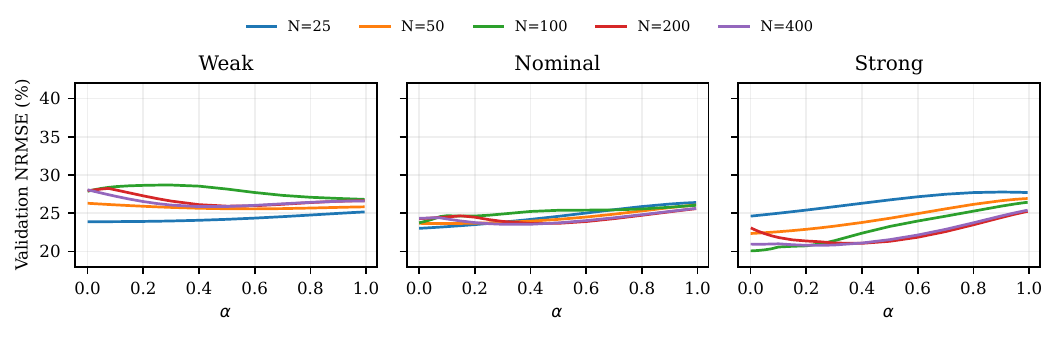}
\caption{Koopman-GL-M validation NRMSE versus $\alpha$ for several retained
memory horizons. For each $(N,\alpha)$ pair, ridge regularization is
selected using validation data. No test data are used for model selection.}
\label{fig:alpha}
\end{figure*}

The selected memory-model hyperparameters are reported in
Table~\ref{tab:selection}, where W, NO, and S denote weak,
nominal, and strong memory regimes, respectively. The selected models are then frozen before
evaluation on the test set.

\begin{table}[t]
\caption{Selected memory-model hyperparameters.}
\label{tab:selection}
\centering
\small
\begin{tabular}{llrrr}
\toprule
Model & Regime & $N$ & $\alpha$ & $\lambda$\\
\midrule
State-GL-M
    & W / NO & 400 & 0.075 & $10^{-3}$\\
    & S     & 400 & 0.1   & $10^{-3}$\\
Koopman-GL-M
    & W     & 25  & 0.03  & $10^{-4}$\\
    & NO     & 25  & 0.005 & $10^{-4}$\\
    & S     & 100 & 0.005 & $10^{-6}$\\
Koopman-Free-Lag
    & W / NO / S & 400 & --- & $10^{-4}$\\
\bottomrule
\end{tabular}
\end{table}

\subsection{Finite-Horizon Open-Loop Prediction}

With all model parameters and hyperparameters fixed from the training and
validation stages, we evaluate the identified predictors on the held-out
test set using finite-horizon open-loop prediction. Each rollout is
initialized from the measured history and propagated using the recorded
future torque sequence, with no subsequent state correction.

Let $\sigma_i$ denote the standard deviation of physical state $i$
computed from the training data and held fixed for all model comparisons. For $n_t$ test trajectories, prediction accuracy (\%) at horizon $H$
is measured by
\begin{equation*}
\mathrm{NRMSE}_H
=
100\times
\sqrt{
\frac{1}{4n_tH}
\sum_{\ell=1}^{n_t}
\sum_{k=1}^{H}
\sum_{i=1}^{4}
\left(
\frac{x_{\ell,k,i}-\widehat x_{\ell,k,i}}{\sigma_i}
\right)^2
}.
\label{eq:metric}
\end{equation*}
Table~\ref{tab:performance} reports the resulting test errors at
$H=100$ and $H=300$, corresponding to prediction horizons of 2 s and
6 s, respectively.

\begin{table*}[t]
\caption{Test aggregate NRMSE (\%) for finite-horizon open-loop prediction.}
\label{tab:performance}
\centering
\input{performance_table.tex}
\end{table*}

The central comparison is between Koopman-Markov and Koopman-GL-M, since
the two models use the same lifted representation and differ primarily in
the inclusion of the structured history term. At the two-second horizon,
the memory term provides little benefit when the physical memory is weak,
with Koopman-GL-M exhibiting approximately $9\%$ higher error than
Koopman-Markov. In contrast, Koopman-GL-M reduces the prediction error by
approximately $43.5\%$ in the nominal-memory regime and $50.1\%$ in the
strong-memory regime. The increasing benefit of the memory-augmented model
is consistent with the growing contribution of the hidden relaxation
dynamics as $\gamma$ increases. These results support the main premise of
the proposed formulation: once nonlinear state dependence is represented
in lifted coordinates, explicitly accounting for hereditary effects can
substantially improve prediction, especially when hereditary effects
are dynamically significant.

A similar trend is visible in the state-based predictors, where the
benefit of State-GL-M over State-Markov grows as the memory strength
increases. Because these models operate directly in the measured state
coordinates, their prediction errors also include larger mismatch from the
nonlinear plant dynamics. The corresponding Koopman models reduce this
additional source of error through the nonlinear lifting.

Koopman-Free-Lag provides a higher-flexibility memory baseline. At
$H=100$, it achieves lower test error than Koopman-GL-M in all three
regimes, showing the predictive advantage that can be obtained by fitting
the lag operators independently. This comparison is important because the
proposed structure is not intended to reproduce the flexibility of an
unrestricted lag model without cost. Rather, it restricts that flexibility
to obtain a substantially more compact identification problem, whose
consequences are examined later.

At the longer horizon $H=300$, Koopman-GL-M attains lower test error than
Koopman-Free-Lag in all three regimes. Relative to Koopman-Markov, its
advantage remains pronounced in the strong-memory regime, where the error
is reduced by approximately $37.7\%$, while the two models perform nearly
identically in the nominal regime and Koopman-Markov remains better in the
weak regime. 

To complement the aggregate test results, Fig.~\ref{fig:trajectory} shows
a representative six-second open-loop prediction for one held-out
strong-memory scenario. The four plotted variables are the measured states
$q_1$, $q_2$, $\dot q_1$, and $\dot q_2$; the hidden relaxation states are
not available to any of the identified predictors. All models are
initialized from the same measured history and propagated using the same
recorded torque sequence.

\begin{figure*}[t]
\centering
\includegraphics[width=\textwidth]{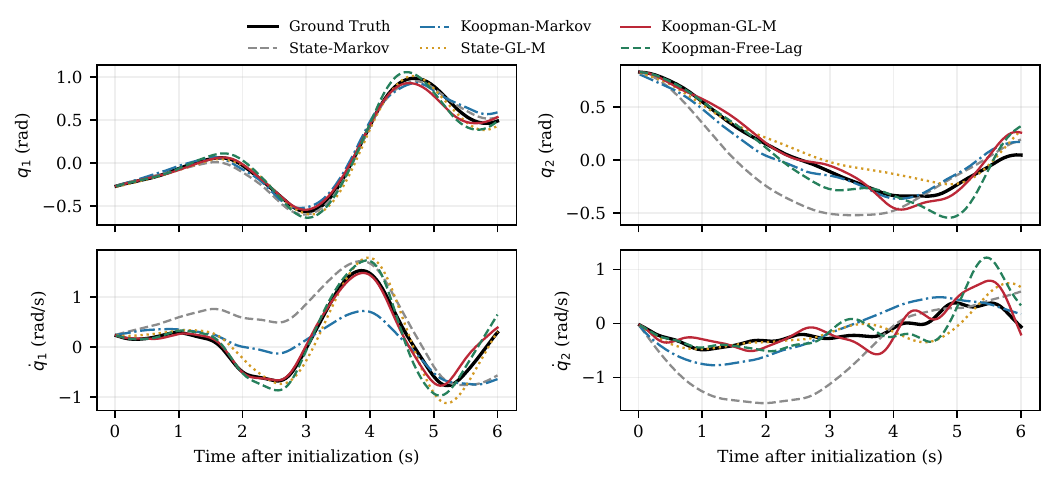}
\caption{Representative six-second open-loop prediction for a held-out
strong-memory test scenario. All predictors are initialized from the same
measured history and propagated using the same recorded future torque
sequence without state correction. Only the four measured physical states
are shown; the hidden relaxation states are not provided to the identified
models.}
\label{fig:trajectory}
\end{figure*}

\subsection{Parameterization and Data Dependence}

The prediction results above expose the main tradeoff between Koopman-GL-M and
Koopman-Free-Lag. With the complete training set, the unrestricted lag model can
achieve lower short-horizon prediction error by fitting each lag operator
independently. This flexibility, however, causes the number of fitted
coefficients to grow with the retained memory horizon. In contrast, Koopman-GL-M
represents all delayed operators through the same matrix $M$ and therefore
retains a fixed number of fitted coefficients as $N$ increases.

For the $p=14$ lifting used here, Koopman-GL-M contains 434 fitted scalar
coefficients for any memory horizon. Koopman-Free-Lag contains
$196N+42$ coefficients and therefore requires 78,442 coefficients at the
selected value $N=400$. Thus, our Koopman-GL-M replaces a large
collection of independently fitted lag matrices with a substantially
smaller parameterization while retaining competitive finite-horizon
prediction.

The practical consequence of this reduction becomes apparent when the
amount of identification data is reduced. Fig.~\ref{fig:data} examines
this effect in the nominal-memory regime at $H=100$ using two complementary
experiments. In Fig.~\ref{fig:data}(a), the memory structures selected
using the complete training set are held fixed as the number of training
trajectories decreases. Thus, Koopman-GL-M retains $N=25$ and Koopman-Free-Lag retains
$N=400$, while the ridge parameter is reselected at each training size.
Both training and test errors are shown.

\begin{figure*}[t]
\centering
\includegraphics[width=\textwidth]{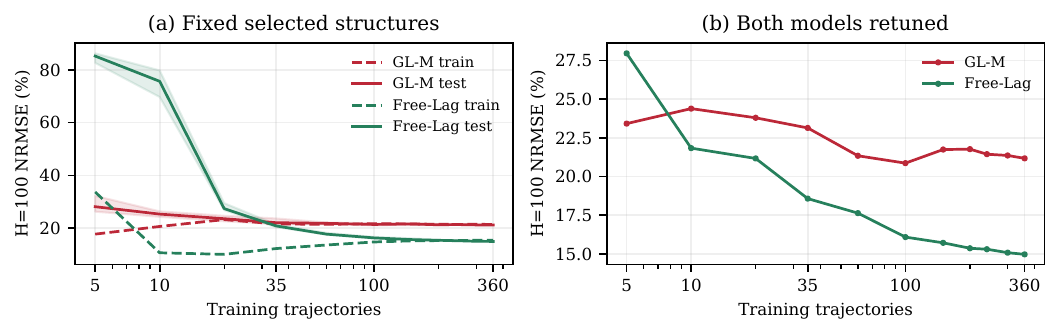}
\caption{Dependence on the number of training trajectories in the
nominal-memory regime at $H=100$. (a) The memory structures selected
using the complete training set are held fixed while ridge regularization
is reselected; training and test errors are reported over ten subset
replicates. (b) Both Koopman-GL-M and Koopman-Free-Lag reselect all
applicable hyperparameters at every training size; median test errors are
reported over three replicates.}
\label{fig:data}
\end{figure*}

Under limited data, the high-capacity Koopman-Free-Lag model develops a
large separation between training and test error, while the corresponding
gap for Koopman-GL-M remains much smaller. This behavior is consistent
with overfitting of the fixed Koopman-Free-Lag parameterization when the
available data are insufficient to support its large number of fitted
coefficients. As the training set grows, this gap diminishes and
Koopman-Free-Lag recovers its short-horizon accuracy.

Fig.~\ref{fig:data}(b) gives Koopman-Free-Lag a more favorable
comparison by allowing both model families to retune their memory
structure at every training size. Koopman-Free-Lag can then reduce its
identification burden by selecting a shorter history. Even under this
retuning, Koopman-GL-M achieves lower test error at the smallest tested
training size, while Koopman-Free-Lag becomes more accurate once additional
data are available.

Together, the two experiments illustrate the central tradeoff of the
proposed parameterization. Koopman-Free-Lag offers greater flexibility
when sufficient data are available to estimate its larger set of lag
coefficients, whereas Koopman-GL-M imposes a shared temporal structure that
substantially reduces the identification burden while retaining competitive
prediction accuracy. The benefit of Koopman-GL-M is therefore not that it
dominates an unrestricted lag model in every data regime, but that it
provides a more compact memory representation whose performance degrades
less severely when the available identification data become limited.

\subsection{Scope and Limitations}

The numerical study considers a single noise-free nonlinear benchmark with
a fixed lifting dictionary and feedback-generated trajectories. The
data-dependence results therefore demonstrate an empirical
parameterization tradeoff rather than a general sample-complexity
guarantee. In addition, the identified predictors are not constrained to
be stable, and the present study does not address measurement noise or
closed-loop performance. These questions remain important directions for
future work.

\section{Conclusion}

This paper developed Koopman-GL-M, a structured Koopman-lifted
finite-memory predictor for nonlinear hereditary systems. The proposed
model represents delayed effects through a shared normalized GL temporal
profile and a learned coupling matrix, yielding a number of fitted
coefficients that is independent of the retained memory horizon. For fixed
temporal hyperparameters, identification remains a linear ridge-regression
problem. The accompanying analysis characterizes the effects of temporal
profile mismatch, memory truncation, residual lifted-model error, and
regularization, without assuming an exact GL memory law or exact
finite-dimensional Koopman closure.

On a nonlinear two-link manipulator with hidden viscoelastic memory,
Koopman-GL-M substantially improves prediction over Koopman-Markov when the
physical memory is nominal or strong. Koopman-Free-Lag achieves lower
short-horizon error with abundant data, but requires 78,442 fitted
coefficients compared with 434 for Koopman-GL-M. The data-dependence study
shows the practical consequence of this reduction: the compact
Koopman-GL-M parameterization remains competitive as the amount of
identification data decreases, whereas the high-capacity
Koopman-Free-Lag model becomes more difficult to identify reliably.
Overall, the results demonstrate the tradeoff enabled by Koopman-GL-M
between predictive flexibility and identification complexity. Future work
will consider stability-constrained fitting, measurement noise, and
closed-loop validation.

\section*{AI Assistance Disclosure}
AI tools were used for mathematical checking and manuscript drafting. The authors are responsible for the scientific claims and numerical provenance.

\bibliographystyle{ieeetr}
\bibliography{refs}
\end{document}

%% file: performance_table.tex
\begin{tabular}{lrrrrrr}\toprule
& \multicolumn{3}{c}{$H=100$ (2 s)} & \multicolumn{3}{c}{$H=300$ (6 s)}\\
Model & Weak & Nominal & Strong & Weak & Nominal & Strong\\\midrule
State-Markov & 36.27 & 38.89 & 46.35 & 44.61 & 45.06 & 60.12\\
Koopman-Markov & 20.69 & 37.45 & 36.27 & \textbf{26.10} & 41.56 & 40.13\\
State-GL-M & 31.86 & 29.70 & 27.05 & 82.80 & 75.35 & 65.63\\
Koopman-GL-M & 22.55 & 21.17 & 18.10 & 39.37 & \textbf{41.17} & \textbf{25.01}\\
Koopman-Free-Lag & \textbf{15.97} & \textbf{14.96} & \textbf{13.77} & 56.37 & 51.14 & 46.00\\
\bottomrule\end{tabular}